\pdfoutput=1 

\documentclass[sigplan, screen, 10pt, authorversion]{acmart}

\setcopyright{acmlicensed}
\acmPrice{15.00}
\acmDOI{10.1145/3372885.3373820}
\acmYear{2020}
\copyrightyear{2020}
\acmISBN{978-1-4503-7097-4/20/01}
\acmConference[CPP '20]{Proceedings of the 9th ACM SIGPLAN International Conference on Certified Programs and Proofs}{January 20--21, 2020}{New Orleans, LA, USA}
\acmBooktitle{Proceedings of the 9th ACM SIGPLAN International Conference on Certified Programs and Proofs (CPP '20), January 20--21, 2020, New Orleans, LA, USA}

\usepackage[utf8]{inputenc}
\usepackage[T1]{fontenc}
\usepackage{microtype}

\usepackage{booktabs}   
\usepackage{subcaption} 

\acmConference{}{}{} 
\setcopyright{none} 

\usepackage{balance}
\usepackage{hyperref}
\usepackage{float}
\usepackage{amssymb}
\usepackage{amsmath}
\usepackage{stmaryrd}
\usepackage{xspace}
\usepackage{mathpartir}
\usepackage{tikz}
\usepackage{soul} 
\usetikzlibrary{cd}

\usetikzlibrary{arrows,arrows.meta}

\tikzcdset{arrow style=tikz, diagrams={>= stealth }}

\newcommand{\formatrule}[1]{(\textit{#1})\xspace}
\newcommand{\apply}{\formatrule{apply}}
\newcommand{\ifTrue}{\formatrule{ifTrue}}
\newcommand{\ifFalse}{\formatrule{ifFalse}}
\newcommand{\new}{\formatrule{new}}
\newcommand{\get}{\formatrule{get}}
\newcommand{\set}{\formatrule{set}}
\newcommand{\fork}{\formatrule{fork}}
\newcommand{\join}{\formatrule{join}}
\newcommand{\joinepsilon}{\formatrule{join$_\epsilon$}}

\newcommand{\leftto}{\leftarrow}
\newcommand{\toeq}{\to_\abeq}
\newcommand{\lefttoeq}{\leftto_\abeq}
\DeclareSymbolFont{lasy}{U}{lasy}{m}{n}
\SetSymbolFont{lasy}{bold}{U}{lasy}{b}{n}
\DeclareMathSymbol\leadsto{\mathrel}{lasy}{"3B}
\newcommand{\theory}[1]{\href{https://www.isa-afp.org/browser_info/current/AFP/Concurrent_Revisions/#1.html}{\texttt{#1.thy}}\xspace}
\newcommand{\data}{\theory{Data}}
\newcommand{\occurrences}{\theory{Occurrences}}
\newcommand{\renaming}{\theory{Renaming}}
\newcommand{\substitution}{\theory{Substitution}}
\newcommand{\operationalsemantics}{\theory{OperationalSemantics}}
\newcommand{\executions}{\theory{Executions}}
\newcommand{\determinacy}{\theory{Determinacy}}

\newcommand\bnfmid{\;\mid\;}
\newcommand{\abeq}{\mathcal{C}}
\newcommand{\E}{\mathcal{E}}
\newcommand{\R}{\mathcal{R}}

\newcommand{\V}{\mathcal{V}}
\newcommand{\rfork}{\mathsf{rfork}}
\newcommand{\rjoin}{\mathsf{rjoin}}
\newcommand{\unit}{\mathsf{unit}}
\newcommand{\Ref}{\mathsf{ref}}
\newcommand{\Read}{\mathsf{!}}
\newcommand{\dom}{\mathsf{dom}}
\newcommand{\doms}{\mathsf{doms}}
\newcommand{\ran}{\mathsf{ran}}
\newcommand{\LID}{\textit{LID}}
\newcommand{\RID}{\textit{RID}}
\newcommand{\id}{\mathit{id}}
\newcommand{\subsume}{\mathcal{S}}
\newcommand{\subsnap}{\mathcal{A}}
\newcommand{\happrox}{\hspace{.6mm}\rotatebox[origin=c]{270}{$\approx$}}

\begin{document}
	
\title{Formalizing Determinacy of Concurrent Revisions}

\author{Roy Overbeek}
\affiliation{
	\department{Department of Computer Science}              
	\institution{Vrije Universiteit Amsterdam}            
	\city{Amsterdam}
	\country{The Netherlands}                    
}
\email{r.overbeek@vu.nl}          



\begin{abstract}
	Concurrent revisions is a concurrency control model designed to guarantee \emph{determinacy}, meaning that the outcomes of programs are uniquely determined. This paper describes an Isabelle\slash HOL formalization of the model's operational semantics and proof of determinacy. We discuss and resolve subtle ambiguities in the operational semantics and simplify the proof of determinacy. Although our findings do not appear to correspond to bugs in implementations, the formalization highlights some of the challenges involved in the design and verification of concurrency control models.
\end{abstract}

\maketitle

\section{Introduction}
\label{sec:introduction}

\emph{Concurrency control models} provide abstractions that simplify the task of writing concurrent software. Such abstractions may assure the programmer, for instance, that intermediate program states of a process are not visible to other processes (\emph{isolation}), or that blocks of instructions execute as a single indivisible unit (\emph{atomicity}). These assumptions simplify reasoning about a program's behavior and prevent undesirable interactions between processes.

\emph{Concurrent revisions} (\emph{CR}) is a concurrency control model originally published by Burckhardt et al.\ in 2010~\cite{burckhardt2010concurrent}. Unlike the relatively established family of \emph{transactional memory} (\emph{TM})~\cite{herlihy1993transactional,shavit1995software} models, which take inspiration from database transactions, the design of CR is modeled after branching version control systems such as Git. This unorthodox starting point gives rise to some distinguishing features, including:
\begin{itemize}
	\item \emph{Non-linear program state history}. In traditional concurrent programming models, it makes sense to speak of `the' state of shared data. Any local views that processes have of this state may be considered deviating, e.g., because they are stale, or because an update is being prepared locally. By contrast, in CR there is no such singular shared state: there exists only the collection of local views on shared data.
	\item \emph{Deterministic conflict resolution}. Processes must sometimes converge, while their local views may conflict. Rather than issuing rollbacks in the event of conflict (as in TM), in CR the conflict is resolved at run time using \emph{deterministic merge functions}. Which merge function to apply is context-dependent, and is declaratively defined by the programmer using semantic type annotations.
	\item \emph{Determinacy}. A concurrency control model is \emph{determinate} if the outcome of programs is guaranteed to be uniquely determined~\cite{karp1966properties}. Most models are not determinate, since scheduling may influence a program's outcome. For instance, the outcome of a lock-based approach may depend on which thread first acquires a particular lock. For TM, the outcome may depend on which transaction is successfully committed first. By contrast, any CR program (satisfying some simple conditions) is determinate, regardless of asynchronous execution and scheduling. This simplifies the life of the programmer, who no longer needs to reason about the timing of events.
\end{itemize}

CR has been implemented in C\# by Burckhardt et al.~\cite{burckhardt2010concurrent}. This implementation is accompanied by a case study in the form of a game implementation, for which a considerable speedup is observed relative to a sequential version, and the corresponding code is arguably easy to reason about. A Haskell implementation followed later by Leijen et al.~\cite{leijen2011prettier}. The implementations are supported by a formal operational semantics by Burckhardt and Leijen \cite{burckhardt2011semantics} (supplemented with a relevant technical report \cite{burckhardt2010semanticstech}), which contains a proof of determinacy as one of its central results.

Concurrency control models, being intricate pieces of concurrent software, are generally interesting targets for formal specification and verification. There are numerous formal approaches to the family of TM models~\cite{harris2005composable, abadi2008semantics, cohen2008mechanical, doherty2013towards, doherty2017proving}, for instance, and some of these efforts uncovered bugs in popular models that lead to fixes in existing software libraries~\cite{manovit2006testing}. The operational semantics of concurrent revisions, however, has not yet been formalized. 

This paper contributes the first step towards the formal verification of CR, using the formal operational semantics of \cite{burckhardt2011semantics} (henceforth referred to as the ``original account'') as our basis. The formalization was performed using the proof assistant \mbox{Isabelle\slash HOL} \cite{nipkow2002isabelle}. Our main results are
\begin{itemize}
	\item the identification and resolution of subtle ambiguities in the side conditions of the rules of the operational semantics, resulting in the strengthening of a side condition and the elimination of three redundant side conditions; and
	\item the mechanization and simplification of the proof of determinacy, in which we show that the proof relies on a property not mentioned in the original account.
\end{itemize}
The verification of an orthogonal desired property, namely, the existence of unique greatest common ancestors in revision diagrams \cite{burckhardt2011semantics} (the meaning of which will become clearer in Section 2), is left for future work.

The formalization artifact is available at the Archive of Formal Proofs \cite{overbeek2018formalization} and consists of about 3000 lines of Isabelle code. More details can be found in the author's master's thesis~\cite{overbeek2018thesis}.

In the remainder of this paper, we first provide an overview of CR and describe its formal semantics (Section~\ref{sec:concurrent-revisions}). Then, we explain the formalization in three parts, covering respectively preliminary aspects (Section~\ref{sec:formalization-preliminaries}), the operational semantics (Section~\ref{sec:operational-semantics}) and the proof of determinacy (Section~\ref{sec:determinacy}). Finally, we discuss the significance of our findings, in part by considering CR implementations and related work (Section~\ref{sec:discussion-and-related-work}).
\section{Concurrent Revisions}
\label{sec:concurrent-revisions}

In this section we first give an informal, high-level overview of CR (Section \ref{sec:overview}) exhibiting the central ideas. Then, we systematically describe and comment on the formal semantics as defined in the original account (Section~\ref{sec:formal-semantics}). 

\subsection{Overview}
\label{sec:overview}

The central unit of concurrency in the CR model is the \emph{revision}. A revision can be thought of as a process evaluating an expression $e$ using a (conceptually) isolated, local \emph{store} $\gamma = \{ l_1 \mapsto v_1, \ldots, l_n \mapsto v_n \}$, which maps \emph{locations} $l_i$ to \emph{values} $v_i$. A revision is uniquely identified by an \emph{identifier}. All computation within the model takes place within some revision. Initially, there is only one revision called the \emph{main revision}. We write $\{ r \mapsto \langle \gamma, e \rangle \}$ to denote a program state in which revision $r$ evaluates $e$ using store $\gamma$. 

Revisions execute in complete isolation from one another, unless an explicit synchronization operation -- fork or join -- is performed.

When a revision $r_1$ \emph{forks} some expression $e$, a fresh revision $r_2$ is created that evaluates $e$. Revision $r_2$ is initialized with a copy of $r_1$'s store (a \emph{snapshot}), and the identifier $r_2$ is exposed to $r_1$. Let $\E[e]$ denote an expression where $\E[ \ ]$ represents an evaluation context around $e$. Then
\[
\{r_1 \mapsto \langle \gamma, \E[\rfork \ e] \rangle \} \to
\{r_1 \mapsto \langle \gamma, \E[r_2] \rangle, r_2 \mapsto \langle \gamma, e \rangle \}
\]
represents an example in which $r_1$ forks $e$. (Informally, we also say that $r_1$ forks $r_2$.)

When revision $r_1$ has a reference to $r_2$, then $r_1$ can \emph{join} $r_2$. This causes $r_1$ to block until $r_2$ terminates. Once $r_2$ terminates, the store of $r_2$ is \emph{merged} into $r_1$'s store, and $r_2$ ceases to exist. Joining a nonexistent revision is considered an error. If $e$ is in normal form (signifying termination of $r_2$), then
\[
\begin{array}{llll}
\{r_1 \mapsto \langle \gamma_1, \E[\rjoin \ r_2] \rangle, r_2 \mapsto \langle \gamma_2, e \rangle \} \to_{r_1} \\
\{r_1 \mapsto \langle \mathcal{M}(\gamma_1, \gamma_2), \E[ \unit ] \rangle \}
\end{array}
\]
represents an example in which $r_1$ joins $r_2$, with $\mathcal{M}$ representing the merge function.

To explain how the merge function $\mathcal{M}$ works, we first introduce the notion of a \emph{revision diagram}, which visualizes the interactions between revisions. In these diagrams, solid arrows depict steps within revisions, and dotted arrows depict fork and join relations between revisions. The following is a simple example, in which four states are labeled:
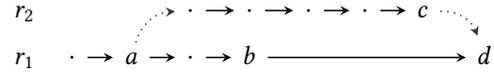
\begin{center}
		\begin{tikzcd}[row sep=0.2em, column sep=1em, fork/.style={dotted, bend left=35pt, opacity=0.8}, join/.style={dotted, bend left=35pt, opacity=0.8}] 
	r_2 \hspace{-2mm} & &  & \cdot \arrow[r]  & \cdot \arrow[r]  & \cdot \arrow[r]  & \cdot \arrow[r]  & c \arrow[rd, join]  &  &  \\
	r_1 \hspace{-2mm} & \cdot \arrow[r]  & a \arrow[r]  \arrow[ru, fork]  & \cdot \arrow[r]  & b \arrow[rrrr]  &  &   &  & d
\end{tikzcd}
\end{center}
In state $a$, main revision $r_1$ forks $r_2$. In state $b$, $r_1$ initiates a join on $r_2$, which blocks until $r_2$ reaches its terminal state $c$. State $d$ is the result of $r_1$ joining $r_2$. State $a$ is the \emph{greatest common ancestor} (\emph{gca}) of joiner state $b$ and joinee state $c$. (The initial state is regarded as the minimal element). Burckhardt and Leijen have shown that each pair of states $(x,y)$
has a unique gca: see Lemma 17 and Theorem 10 of the technical report~\cite{burckhardt2010semanticstech}.

Let $x_\gamma$ denote the store at a state $x$, and $\mathcal{W}(x,y)$ the set of locations that were written to in the execution from state $x$ to state $y$. The merge $\mathcal{M}$ of stores $b_\gamma$ (belonging to a joining revision $r_1$) and $c_\gamma$ (belonging to a joined revision $r_2$) with gca store $a_\gamma$ (see the diagram above) is defined as follows:
\[
\mathcal{M}(b_\gamma, c_\gamma) \ l = 
\begin{cases}
b_\gamma \ l &  \hspace{-5.95pt} l \notin \mathcal{W}(a, c)\\
c_\gamma \ l & \hspace{-5.95pt}l \in \mathcal{W}(a, c) \land l \notin \mathcal{W}(a, b) \\
f_l(a_\gamma \ l, b_\gamma \ l, c_\gamma \ l)& \hspace{-5.95pt} \text{otherwise}
\end{cases}
\]
Here, $f_l$ is a deterministic \emph{merge function} that resolves the \emph{write-write} conflict on $l$. It is uniquely determined by the \emph{isolation type} of $l$: a user-definable type for shared locations that describes how conflicts should be resolved.

We illustrate the concept of an isolation type using two standard examples: the \textit{Versioned} and \textit{Cumulative} isolation types.

If $l$ stores a \textit{Versioned} integer, then $f_l(v_1,v_2,v_3) = v_3$, effectively prioritizing the joinee and possibly overwriting a modification by the joiner. This behavior is illustrated by the following revision diagram:
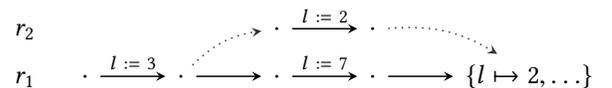
\begin{center}
\begin{tikzcd}[row sep=0.2em, column sep=2.5em, fork/.style={dotted, bend left=18pt, opacity=0.8}, join/.style={dotted, bend left=18pt, opacity=0.8}] 
	r_2 \hspace{-6mm} &                                        &                                                      & \cdot \arrow[r, "l \ := \ 2"	]  & \cdot \arrow[rd, join]  \\
	r_1 \hspace{-6mm} & \cdot \arrow[r, "l \ := \ 3"]  & \cdot \arrow[r]  \arrow[ru, fork]  & \cdot \arrow[r, "l \ := \ 7"]  & \cdot \arrow[r]  & {\{l \mapsto 2, \ldots\} }
\end{tikzcd}
\end{center}
A datum can be declared \textit{Versioned}, for instance, when the joinee is performing some task enjoying higher priority than the joiner's task.

If $l$ stores a \textit{Cumulative} integer, by contrast, then the merge function is $f_l(v_1,v_2, v_3) = v_2 + v_3 - v_1$, taking both modifications into account. In the following diagram, both revisions added $2$ to the original value of $3$, causing the result of the merge to be $7$:
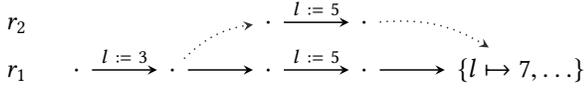
\begin{center}
\begin{tikzcd}[row sep=0.2em, column sep=2.5em, fork/.style={dotted, bend left=18pt, opacity=0.8}, join/.style={dotted, bend left=18pt, opacity=0.8}] 
  r_2 \hspace{-6mm} &  &  & \cdot \arrow[r, "l \ := \ 5"] & \cdot \arrow[rd, join]  \\
  r_1 \hspace{-6mm} & \cdot \arrow[r, "l \ := \ 3"]  & \cdot \arrow[r]  \arrow[ru, fork]  & \cdot \arrow[r, "l \ := \ 5"]  & \cdot \arrow[r]  & {\{l \mapsto 7, \ldots \} }
\end{tikzcd}
\end{center}
A typical use case for the \textit{Cumulative} isolation type is one in which $l$ functions as a counter.

Since identifiers can be exchanged through fork and join operations, valid revision diagrams can be quite complex:
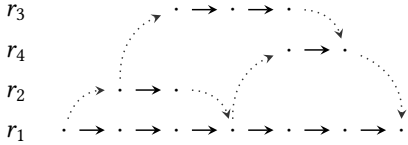
\begin{center}
\begin{tikzcd}[row sep=0.2em, column sep=1em, fork/.style={dotted, bend left=35pt, opacity=0.8}, join/.style={dotted, bend left=35pt, opacity=0.8}] 
	r_3 \hspace{-2mm} &  &  & \cdot \arrow[r]  & \cdot \arrow[r]  & \cdot \arrow[rd, join]  &  &  \\
	r_4 \hspace{-2mm} &  &  &  &  & \cdot \arrow[r]  & \cdot \arrow[rdd, join]  &  \\
	r_2 \hspace{-2mm} &  & \cdot \arrow[r]  \arrow[ruu, fork]  & \cdot \arrow[rd, join]  &  &  &  &  \\
	r_1 \hspace{-2mm}  & \cdot \arrow[r]  \arrow[ru, fork]  & \cdot \arrow[r]  & \cdot \arrow[r]  & \cdot \arrow[ruu, fork]  \arrow[r]  & \cdot \arrow[r]  & \cdot \arrow[r]  & \cdot
\end{tikzcd}
\end{center}
Despite this, programs are \emph{determinate}, meaning that the outcome of a program is uniquely determined, even if scheduling is nondeterministic. This property assumes two simple conditions: (1) revisions do not perform nondeterministic behavior that affects the semantics of outcomes (e.g., generating a random number), and (2) revisions are joined only once (a second join operation would be undefined).

\subsection{Formal Semantics}
\label{sec:formal-semantics}

\begin{figure*}
	$
	\begin{array}{lllll} 
	(\textit{apply}) & s\llbracket r \mapsto \langle \sigma, \tau, \E[ (\lambda x. e) \ v ]   \rangle \rrbracket    & \to_r & s(r \mapsto \langle \sigma, \tau, \E[ [v/x] e] \rangle)  \\
	(\textit{if-true}) & s\llbracket r \mapsto \langle \sigma, \tau, \E[ \mathsf{true} \ \mathsf{?} \ e_1 \ \mathsf{:} \ e_2]\rangle \rrbracket & \to_r & s(r \mapsto \langle \sigma, \tau, \E[ e_1]\rangle) \\
	(\textit{if-false}) & s\llbracket r \mapsto \langle \sigma, \tau, \E[ \mathsf{false} \ \mathsf{?} \ e_1 \ \mathsf{:} \ e_2]\rangle \rrbracket    & \to_r & s(r \mapsto \langle \sigma, \tau, \E[ e_2]\rangle)  \\
	\\
	(\textit{new}) & s\llbracket r \mapsto \langle \sigma, \tau, \E[ \mathsf{ref} \ v]\rangle \rrbracket   & \to_r & s(r \mapsto \langle \sigma, \tau(l \mapsto v), \E[ l ] \rangle)  & \mathsf{if} \ l \notin s \\
	(\textit{get}) & s\llbracket r \mapsto \langle \sigma, \tau, \E[ \mathsf{!} l]\rangle \rrbracket   & \to_r & 
	s(r \mapsto \langle \sigma, \tau, \E[ (\sigma \mathsf{::} \tau) \ l] \rangle)  & \mathsf{if} \ l \in \dom \ (\sigma \mathsf{::} \tau)\\
	(\textit{set}) & s\llbracket r \mapsto \langle \sigma, \tau, \E[ l := v] \rangle \rrbracket & \to_r &
	s(r \mapsto \langle \sigma, \tau(l \mapsto v), \E[ \mathsf{unit}] \rangle) & \mathsf{if} \ l \in \dom \ (\sigma \mathsf{::} \tau) \\
	\\
	(\textit{fork}) & s\llbracket r \mapsto \langle \sigma, \tau, \E[ \mathsf{rfork} \ e]\rangle \rrbracket   & \to_r & 
	s(r \mapsto \langle \sigma, \tau, \E[ r'] \rangle, r' \mapsto \langle \sigma \mathsf{::} \tau, \epsilon, e \rangle)  & \mathsf{if} \ r' \notin s \\
	(\textit{join}) & s\llbracket r \mapsto \langle \sigma, \tau, \E[ \mathsf{rjoin} \ r']\rangle, r' \mapsto \langle \sigma', \tau', v \rangle \rrbracket & \to_r &
	s(r \mapsto \langle \sigma, \tau \mathsf{::} \tau', \E[ \mathsf{unit}] \rangle, r' \mapsto \bot)  \\
	(\textit{join}_\epsilon) & s\llbracket r \mapsto \langle \sigma, \tau, \E[ \mathsf{rjoin} \ r'] \rangle, r' \mapsto \bot\rrbracket  & \to_r & \epsilon  \\
	\end{array}
	$
	\caption{The rules of the operational semantics.}
	\label{fig:operational_semantics}
\end{figure*}

The CR semantics is modeled by the \emph{revision calculus}, which consists of a programming language for revisions, a set of evaluation contexts, notions of local and global states, and an operational semantics on global states. The original account also introduces an equivalence relation on states and a vocabulary for discussing execution traces.

\subparagraph{Preliminaries}
We write $\dom \ f$ and $\ran \ f$ to denote respectively the domain and range of a partial function $f$, $\epsilon$ for the empty partial function, $f \ x = \bot$ for $x \notin \dom \ f$, and $f(x \mapsto y)$ for the partial function obtained by updating $x$ to $y$ in $f$. For $n > 1$, the expression $f (x_1 \mapsto y_1, \ldots , x_{n+1} \mapsto y_{n+1})$ abbreviates $(f (x_1 \mapsto y_1, \ldots, x_n \mapsto y_n))(x_{n+1} \mapsto y_{n+1})$. For a bijective function $f$, we write $f^{-1}$ to denote its inverse. Given partial functions $f$ and $g$, $f :: g$ is a partial function that maps $x$ to $g \ x$  if $x \in \dom \ g$ and to $f \ x$ otherwise (``$g$ shadows $f$''). For functions $f$ and sets $S$, $f \ ' \ S$ denotes $S$ under the image of $f$, i.e., $\{ f \ x \mid x \in S  \}$. We write $\rightsquigarrow^=$, $\rightsquigarrow^*$ and $\rightsquigarrow^n$ for respectively the reflexive closure, reflexive transitive closure and $n$-fold composition of a relation $\rightsquigarrow$, use mirrored arrows $\leftsquigarrow$ to denote inverse relations, and write $R \circ R'$ for the composition of relations $R$ and $R'$, given by $(x,z) \in R \circ R' \iff \exists y. \ (x,y) \in R \land (y,z) \in R'$.

\subparagraph{Expressions}
The programming language is parameterized by three (typically infinite) sets: variables $x \in \textit{Var}$, revision identifiers $r \in \textit{Rid}$ and location identifiers $l \in \textit{Lid}$. It defines a set of constants $c \in \textit{Const}$, containing elements \textsf{unit}, \textsf{true} and \textsf{false}. The sets of values and expressions are mutually defined as follows:
\[
\begin{array}{lcl}
	v \in \textit{Val} & ::= & c \bnfmid x \bnfmid  l \bnfmid r \bnfmid \lambda x.e \\
e \in \textit{Expr} & ::= & v \bnfmid e \ e \bnfmid e \ \mathsf{?} \ e \ \mathsf{:} \ e \bnfmid \mathsf{ref} \ e \bnfmid  \mathsf{!}e \bnfmid \\
 & & e := e \bnfmid \mathsf{rfork} \ e  \bnfmid  \mathsf{rjoin} \ e\\
\end{array}
\] 
For the properties of interest, we do not need to consider $\lambda$-terms modulo $\alpha$-equivalence. This is fortunate, since $\alpha$-equivalence has a reputation of being challenging to formalize \cite{berghofer2007a, urban2011general}.

In some contexts, we will write $e_1 \bullet e_2$ rather than $e_1 \ e_2$ to improve readability. 

\subparagraph{Evaluation Contexts} 
The following set of evaluation contexts is defined:
\[
\begin{array}{lclcl}
	\E \in \mathit{Cntxt} & ::= & \square \bnfmid \E \ e \bnfmid v \ \E \bnfmid \E \ \mathsf{?} \ e \ \mathsf{:} \ e \bnfmid \mathsf{ref} \ \E \bnfmid \\
	& &  \mathsf{!}\E \bnfmid \E := e \bnfmid l := \E  \bnfmid \mathsf{rjoin} \ \E
	\end{array}
\]
The expression $\E[e]$ denotes the result of \emph{plugging} $e$ into the unique hole ($\square$) of $\E$. Evaluation contexts allow decomposing an expression $e = \E[r]$ into an evaluation site $r$ (a redex) and its surrounding context $\E$, enabling rewriting under contexts. A more detailed explanation of evaluation contexts is provided by Harper~\cite[pp. 44--46]{harper2016practical}.

More strongly for CR, a \emph{unique decomposition} lemma holds: $\E[r] = \E'[r']$ implies $\E = \E'$ and $r = r'$ for redexes $r$ and $r'$. Since the operational semantics matches expressions $e$ against patterns of the form $\E[r]$, the unique decomposition lemma thus guarantees that always a unique redex of $e$ is evaluated. For example, the expression $((\lambda x \ldotp x) \ x) \ ((\lambda y \ldotp y) \ y)$ can match against the pattern $\E[(\lambda x \ldotp x) \ x]$, since $\square \ ((\lambda y \ldotp y) \ y)$ is a valid context. It cannot match against $\E[(\lambda y \ldotp y) \ y ]$, however, since $((\lambda x \ldotp x) \ x) \ \square$ is not a valid context.

Uniqueness of decomposition is claimed, but not demonstrated in the original account. We describe its proof in Section~\ref{sec:formalization-preliminaries}.

\subparagraph{State}
Three notions of state are required: the state of a store, the local state of a revision, and the global state. A \textit{Store} is a partial function $\sigma,\tau \in \mathit{Lid} \rightharpoonup \mathit{Val}$, and a \textit{GlobalState} is a partial function $s \in \mathit{Rid} \rightharpoonup \mathit{LocalState}$.

For technical reasons, the local state of a revision is not a tuple $\langle \gamma, e \rangle$, consisting of a store $\gamma$ and expression $e$, as  informally  described in Section~\ref{sec:overview}. Instead, a local state is a triple $L \in \textit{LocalState} = \textit{Snapshot} \times \textit{LocalStore} \times \textit{Expr}$, where $\textit{Snapshot}$ and $\textit{LocalStore}$ are type synonyms for $\textit{Store}$. To understand why, we note that the gca store, required to define the merge operation, always equals the snapshot (initial store) of the joinee. The diagrams of Section~\ref{sec:overview} provide examples, and its proof is given in the original account (Lemma 18 of the technical report \cite{burckhardt2010semanticstech}). Thus, if a revision~$r'$ preserves the snapshot it inherits from its forker~$r$, while tracking its own updates separately, then the gca store can always be obtained from the local state of $r'$ when $r'$ is joined. In the operational semantics, snapshots are never modified and local stores track updates.

We introduce the notations $L_\sigma$, $L_\tau$ and $L_e$ for respectively the first, second and third component of a local state $L$, and define $\doms \ L = \dom \ L_\sigma \cup \dom \ L_\tau$.

\subparagraph{Occurrences} To avoid ambiguities in our discussion of the operational semantics, we introduce a family of functions not present in the original account. We write $\RID \ e$ to denote the set of all revision identifiers occurring in expression $e$, and $\LID \ e$ to denote the set of all location identifiers occurring in~$e$. We analogously define functions $\RID$ and $\LID$ for contexts. For stores $\sigma$, we define $\RID \ \sigma = \bigcup \RID \ ' \ \ran \ \sigma$ and $\LID \ \sigma = \dom \ \sigma \cup \bigcup \LID \ ' \ \ran \ \sigma$. For local states $L$, we define $\RID \ L = \RID \ L_\sigma \cup \RID \ L_\tau \cup \RID \ L_e$, and similarly for $\LID \ L$. For global states $s$, we define $\RID \ s = \dom \ s\cup \bigcup \RID \ ' \ \ran \ s$ and $\LID \ s = \bigcup \LID \ ' \ \ran \ s$.

\subparagraph{Operational Semantics}
The operational semantics (Figure~\ref{fig:operational_semantics}) defines a transition relation on global states, indexed by the revision~$r$ ``performing'' the step. The left hand side of each rule is of the form $s\llbracket r \mapsto L \rrbracket$, and matches any global state~$s$ for which $s \ r = L$.

The first three rules affect only the expression local to~$r$. The original authors state that rule \apply is deterministic, but otherwise they make no explicit assumptions about the capture-avoiding substitution $[v/x]e$.

The next three rules model store interactions.
The side condition for \new, $l \notin s$, is a notational shorthand expressing	 that ``$l$ does not appear in any snapshot or local store of $s$'' \cite{burckhardt2011semantics}. We believe that
\begin{equation}\label{eq:new_wrong}
l \notin \bigcup \{\LID \ L_\sigma \cup \LID \ L_\tau \mid L \in \ran \ s \} \tag{$SC_\textit{new}$} 
\end{equation}
is the literal interpretation of this informal characterization, rather than the more conservative side condition $l \notin \LID \ s$. We examine how the choice of interpretation influences determinacy in Section \ref{sec:operational-semantics}. Note that \new is nondeterministic.

Like rule \new, rule \fork is nondeterministic: the side condition $r' \notin s$ is meant to express that $r'$ ``is not mapped by~$s$, and does not appear in any snapshot or local store of $s$''~ \cite{burckhardt2011semantics}. We believe that
\begin{equation}
\label{eq:fork_wrong}
r' \notin \dom \ s  \cup  \bigcup \{\RID \ L_\sigma \cup \RID \ L_\tau \mid L \in \ran \ s \}\tag{$SC_\textit{fork}$}
\end{equation}
is the literal interpretation of this sentence, rather than $r' \notin \RID \ s$. In Section \ref{sec:operational-semantics} we will show that \eqref{eq:fork_wrong} leads to nondeterminacy.

The join operation is modeled by rules \join and \joinepsilon.

Rule \join resolves all conflicts according to the \textit{Versioned} isolation type. The restriction to this isolation type is part of the original account, and we adopt it here in order to remain faithful. The original account argues that this rule can be generalized by using a custom merge function
\[
\mathit{merge}_l : \mathit{Val} \times \mathit{Val} \times \mathit{Val} \to \mathit{Val}
\]
defined for the values at each location $l$ of respectively the snapshot, the local store of the joiner and the local store of the joinee. Because locations are randomly allocated in the calculus, we argue that it instead may be better to modify the calculus by introducing subtypes of \textit{Val}, which then determine which merge functions are used \cite{overbeek2018thesis}. In addition, one would have to forbid the definition of merge functions whose results depend on nondeterministic aspects, such as the occurrence of particular location and revision identifiers in argument values. Failure to do so would result in nondeterminacy.

 Rule \joinepsilon ensures that the global state collapses to the empty function when an erroneous join is performed. It is needed to establish determinacy \cite{burckhardt2011semantics}.

\subparagraph{Equivalence} Since location and revision identifiers are allocated nondeterministically, an equivalence relation on structures containing identifiers is introduced. Let $\alpha \in \mathit{Rid} \to \mathit{Rid}$,  $\beta \in \mathit{Lid} \to \mathit{Lid}$ and let $S$ be some structure containing identifiers (expressions, stores, etc.). We write $\R \ \alpha \ \beta \ S$ to denote the structure that results from renaming every identifier in $S$ according to $\alpha$ and $\beta$, and $S \approx_{\alpha\beta} S'$ to express that $\alpha$ and $\beta$ are bijections and  $\R \ \alpha \ \beta \ S = S'$.  Structures $S$ and $S'$ are said to be \emph{renaming-equivalent}, denoted $S \approx S'$, if  $S \approx_{\alpha\beta} S'$ for some $\alpha$ and $\beta$.

\subparagraph{Executions}  The original account defines a \emph{program expression} as ``an expression containing no revision identifiers'', and an \emph{initial state} as a global state of the form $\epsilon(r \mapsto \langle \epsilon , \epsilon, e \rangle)$, with $e$ a program expression and $r \in Rid$.  We contend that the characterization of a program expression can be interpreted as either $\RID \ e = \varnothing$ or as $\RID \ e = \LID \ e = \varnothing$. We choose the latter interpretation, since rules \set and \get would anyway block on manually introduced location identifiers. This is because only identifiers allocated by \new can end up in a store's domain. In addition, using the former interpretation causes nondeterminacy if side condition \eqref{eq:new_wrong} is used \cite{overbeek2018thesis}.

Let $\rightarrow \ =  \bigcup \{ \to_r  \mid r \in \mathit{Rid} \}$.  An \emph{execution} is a sequence $s \to^* s'$ with $s$ an initial state. The execution is \emph{maximal} if there does not exist an $s''$ such that $s' \to s''$, and $e \downarrow s$ expresses that there exists a maximal execution for a program expression $e$ that ends in global state $s$.  Determinacy modulo~$\approx$ thus means that $e \downarrow s$ and $e \downarrow s'$ imply $s \approx s'$.  A state $s'$ is \emph{reachable} if there exists an execution $s \to^* s'$ from an initial state $s$.

We say that a property $P$ is an \emph{execution invariant} if $P \ s$ for all reachable states $s$. A property $P$ is an \emph{inductive invariant} if
\begin{itemize}
  \item $P \ s$ for all initial states $s$, and
  \item for all states $s$ and $s'$, $\ s \to s' \land P \ s \Longrightarrow P \ s'$.
\end{itemize}
 Every inductive invariant is an execution invariant, but not vice versa.
\section{Formalization Preliminaries}
\label{sec:formalization-preliminaries}

\begin{figure*}
	$
	\begin{array}{llll}
	\texttt{top\_redex}: & {\textit{redex} \ e} \Longrightarrow {e \rhd (\square, e)} \\
	\texttt{lapply}: & \lnot \> \textit{redex} \ (e_1  \ e_2) \Longrightarrow e_1 \rhd (\E, r)  \Longrightarrow {e_1  \ e_2 \rhd (\E \ e_2, r)} \\
	\texttt{rapply} : & \lnot \> \textit{redex} \ (v \ e_2) \Longrightarrow e_2 \rhd (\E, r) \Longrightarrow {v \ e_2 \rhd (v \ \E,  r)}  \\
	\texttt{ite} : & {\lnot  \> \textit{redex} \ (e_1 \ \mathsf{?} \ e_2 \ \mathsf{:} \ e_3) \Longrightarrow e_1 \rhd (\E, r)} \Longrightarrow {e_1 \ \mathsf{?} \ e_2 \ \mathsf{:} \ e_3 \rhd (\E \ \mathsf{?} \ e_2 \ \mathsf{:}  \ e_3,  r)} \\
	\texttt{ref} : & {\lnot \> \textit{redex} \ (\Ref \ e) \Longrightarrow e \rhd (\E, r)}  \Longrightarrow {\Ref \ e \rhd (\Ref \ \E, r)}  \\
	\texttt{read} : & {\lnot \> \textit{redex} \ (\Read  e) \Longrightarrow e \rhd (\E, r)} \Longrightarrow  {\Read  e \rhd (\Read  \E, r)}  \\
	\texttt{lassign} : & \lnot \> \textit{redex} \ (e_1 := e_2) \Longrightarrow \ e_1 \rhd (\E, r) \Longrightarrow {e_1 := e_2 \rhd (\E := e_2, r)} \\
	\texttt{rassign} : & \lnot \> \textit{redex} \ (l := e_2) \Longrightarrow  e_2 \rhd (\E, r) \Longrightarrow  {l := e_2 \rhd (l := \E,  r)} \\
	\texttt{rjoin} : & {\lnot \> \textit{redex} \ (\rjoin \ e) \Longrightarrow  e \rhd (\E, r)}  \Longrightarrow {\rjoin \ e \rhd (\rjoin \ \E, r) }
	\end{array}
	$
	\caption{Predicate \texttt{decompose}, which asserts how expressions can be decomposed.}
\end{figure*}

We briefly describe the formalization of all aspects of the semantics that are preliminary to the mechanization of the operational semantics. These aspects are defined in the Isabelle theories \data, \occurrences, \renaming and \substitution. Theory \data imports \texttt{Main}, meaning that it depends only on a standard assortment of Isabelle libraries.

\subparagraph{Data} Theory \data defines the inductive data types \texttt{const}, \texttt{('r,'l,'v) val}, \texttt{('r,'l,'v) expr} and \texttt{('r,'l,'v) cntxt} required for formalizing expressions (Section \ref{sec:formal-semantics}). In the latter three definitions, \texttt{'r}, \texttt{'l} and \texttt{'v} are type parameters for respectively the types of revision identifiers \textit{Rid}, location identifiers \textit{Lid} and variables \textit{Var}. The theory also defines the notions of stores and states, and some of the related notations and operations, such as projection functions for local states. In Isabelle, partial functions $\alpha \rightharpoonup \beta$ are modeled using option types, i.e., as total functions $\alpha \to \beta \ \texttt{option}$.

Theory \data also contains all definitions related to plugging and decomposing. Most notably, it contains the proof of the unique decomposition lemma (formalized as lemma \texttt{completion\_eq}) mentioned in Section \ref{sec:formal-semantics}. The proof for this lemma has the following structure. First, a particular decomposition for terms containing redexes is defined, given in Figure~2, and formalized as inductive predicate \texttt{decompose}. Intuitively, $e \rhd (\E, r)$ is meant to assert that expression $e$ decomposes into context $\E$ and redex $r$. The decomposition is shown to be valid and unique, respectively:
\begin{lemma}[\texttt{plug\_decomposition\_equivalence}]
	For redexes $r$, \hspace{0.5mm}  $e \rhd (\E,r) \iff \E[r] = e$.
	\end{lemma}
\begin{proof}
	Direction $\Longrightarrow$ follows by rule induction on $e \rhd (\E, r)$. Direction $\Longleftarrow$ is shown by structural induction on $\E$.
\end{proof}
\begin{lemma}[\texttt{unique\_decomposition}]
	If $e \rhd (\E_1, r_1)$ and $e \rhd (\E_2, r_2)$, then $\E_1 = \E_2$ and $r_1 = r_2$.
	\end{lemma}
\begin{proof}
	By rule induction on $e \rhd (\E_1, r_1)$.
\end{proof}

Proofs of unique decomposition lemmas have a reputation of being tediously routine and error-prone \cite{xiao2001from}. This is also our experience, and we think the many inductive cases provide some indication for that. Isabelle's \textit{auto} proof method, however, is able to solve all of these cases automatically once configured with the supporting lemma below and (automatically generated) introduction and elimination  rules for \texttt{decompose}.
\begin{lemma}[\texttt{plugged\_redex\_not\_val}]
  If $r$ is a redex, then $\E[r] \notin \textit{Val}$.
\end{lemma}
 
\subparagraph{Occurrences}
Theory \occurrences defines the $\RID$ and $\LID$ definitions for stores, local states and global states. (The $\RID$ and $\LID$ definitions for values, expressions and contexts are automatically introduced with the data type declarations in \data.)

The theory also proves lemmas that are useful for reasoning about occurrences of location and revision identifiers. For instance, suppose we wish to prove $\RID \ v \subseteq \RID \ s(r \mapsto \langle \sigma, \tau(l \mapsto v), \E[e] \rangle)$. Ideally, we would like to automate the proofs to such obvious lemmas  as much as possible. To this end, we prove a number of simplification rules that flatten complex expressions such as
$\RID \ s(r \mapsto \langle \sigma, \tau(l \mapsto v), \E[e] \rangle)$ into simpler ones such as
\[
\begin{array}{lll}
\RID \ s(r \mapsto \bot) \cup  \{r\} \cup \RID \ \sigma \cup \RID \ \tau(l \mapsto \bot) \  \cup \\
\RID \ v \cup \RID \ \E  \cup \RID \ e\text{,}
\end{array}
\]
since Isabelle's automation tools can easily reason about sets. Similarly, we declare a number of introduction and elimination rules for expressions that cannot be flattened. An example is the introduction rule
\[ 
r \in \RID \ (\sigma :: \tau) \Longrightarrow r \notin \RID \ \sigma \Longrightarrow r \in \RID \ \tau \]
named \texttt{ID\_combination\_subset\_union(1)} in the Isabelle formalization.

\subparagraph{Renaming} Theory \renaming contains all of the definitions and laws related to renaming. Like the $\RID$ and $\LID$ functions, the various renaming functions are discriminated using subscripts in Isabelle, which we omit in this paper.

For values $v$, the renaming $\R \ \alpha \ \beta \ v$ is defined as an abbreviation for $\texttt{map\_val} \ \alpha \ \beta \ \id \ v$, where  \texttt{map\_val} is a function automatically generated by the data type declaration of \texttt{val}. Here, $\texttt{map\_val} \ \alpha \ \beta \ \id \ v$ is the value obtained by renaming location identifiers, revision identifiers and variables according to $\alpha$, $\beta$ and the identity function, respectively. Abbreviations are analogously defined for the renaming of expressions and contexts. The renaming of a store $\sigma$, $\R \ \alpha \ \beta \ \sigma$, is formalized as the function
\[
(\R \ \alpha \ \beta \ \sigma) \ l = \sigma \ (\beta^{-1} \ l) \ \text{>>=} \ (\lambda v \ldotp \R \ \alpha \ \beta \ v)\text{,}
\]
where \text{>>=} is the \emph{bind operator} satisfying $(\texttt{None} \ \text{>>=} \ f) = \texttt{None}$ and $(\texttt{Some} \ x  \ \text{>>=} \ f) = \texttt{Some} \ (f \ x)$ for option types. We show that the renaming is well defined for bijections $\beta$ (lemma \texttt{$\R_S$\_implements\_renaming}). The renaming of a global state is defined in a similar fashion, and the renaming of a local state is straightforwardly defined as a renaming of its components.

The relation $\approx$ is defined and established to be an equivalence (lemmas $\alpha\beta\texttt{\_refl}$, $\alpha\beta\texttt{\_sym}$ and $\alpha\beta\texttt{\_trans}$). This requires proving several identity, composition and inverse laws for each of the renaming functions.

We prove several distributive laws that serve as simplification rules for renamings. For instance, the term $\R \ \alpha \ \beta \ (s(r \mapsto \langle \sigma, \tau(l \mapsto v), \E[e] \rangle))$ is configured to simplify to
\[
\begin{array}{lll}
\R \ \alpha \ \beta \  s(\alpha \ r \mapsto \langle \R \ \alpha \ \beta \ \sigma, \R \ \alpha \ \beta \ \tau(\beta \ l \mapsto \R \ \alpha \ \beta \ v), \\
 \hspace{23.2mm} \R \ \alpha \ \beta \ \E[ \R \ \alpha \ \beta \ e] \rangle)\text{.}
 \end{array}
\]

We distinguish a special class of bijective renamings of the form $\mathit{id}(x := y, y := x)$ that we call \emph{swaps}. All renamings used in proofs are swaps. Several rules are proven that help eliminate ``redundant'' swaps. An example of such a rule states that if $l \notin \LID \ v$ and $l' \notin \LID \ v$, then $\R \ \mathit{id} \ \mathit{id}(l := l', l' := l) \ v = v$ (lemma \texttt{eliminate\_swap\_val(2)}). The swap rules are declared as both simplification and introduction rules.

\subparagraph{Substitution} As observed in Section~\ref{sec:formal-semantics}, rule \apply presupposes a notion of substitution, but the original account does not specify which one. For this reason, we also do not fix a particular notion of substitution. Instead, theory \substitution defines a locale called \texttt{substitution}. The locale fixes a constant \texttt{subst}, and introduces three assumptions:
\begin{enumerate}
	\item \texttt{renaming\_distr\_subst}: \\ $\R \ \alpha \ \beta \ (\texttt{subst} \ e \ x \ e') =  \texttt{subst} \ (\R \ \alpha \ \beta \ e) \ x \ (\R \ \alpha \ \beta \ e')$;
	\item \texttt{subst\_introduces\_no\_rids}: \\ $\RID \ (\texttt{subst} \ e \ x \ e')  \subseteq \RID \ e \cup \RID \ e'$; and
	\item \texttt{subst\_introduces\_no\_lids}: \\ $\LID \ (\texttt{subst} \ e \ x \ e')  \subseteq \LID \ e \cup \LID \ e'$.
\end{enumerate}
We found that these assumptions were sufficient for proving determinacy.

We provide two models for \texttt{substitution} that demonstrate that the assumptions are satisfiable. The first is a trivial model, in which \texttt{subst} is interpreted as a constant function that maps to $\unit$:
$\texttt{constant\_function} \ e \ x \ e' = \mathsf{unit}$.
The fact that this constant function is a model (proven in lemma \texttt{constant\_function\_models\_substitution}) indicates that the assumptions on \texttt{subst} are weak.

The second model, function \texttt{nat\_subst$_\texttt{E}$}, is a more faithful instance of a deterministic substitution function in which natural numbers are used as variables. It is mutually recursively defined with \texttt{nat\_subst$_\texttt{V}$}, which implements substitution for values. Let $\V \ e$ denote the set of (free and bound) variables that occur in the expression $e$, and let $e_{x \mapsto y}$ denote the expression obtained by renaming \emph{every} variable $x$ in $e$ to $y$. The following case of the definition illustrates how deterministic capture-avoiding substitution is implemented:
\[
\begin{array}{lll}
\texttt{nat\_subst}_\texttt{V} \ e \ x \ (\lambda y\ldotp e') = \\
\hspace{10mm}
\begin{cases}
 \lambda y\ldotp e' & \text{if $x = y$}\\
\lambda z \ldotp \texttt{nat\_subst$_\texttt{E}$} \ e \ x \ e'_{y \mapsto z} & \text{otherwise}
\end{cases}
\end{array}
\]
where $z = \textit{max}(\V \ e \ \cup \ \V \ e') + 1$. For further technical details, such as why bound variables are also renamed, we refer to the author's master's thesis~\cite{overbeek2018thesis}.

\section{Operational Semantics}
\label{sec:operational-semantics}

We are now ready to formalize the operational semantics. Recall from Section \ref{sec:formal-semantics} that we have to choose between the \fork side conditions \eqref{eq:fork_wrong} and $r \notin \RID \ s$, and between the \new side conditions \eqref{eq:new_wrong} and $l \notin \LID \ s$. In this section, we first show that \eqref{eq:fork_wrong} is too weak, since it leads to an indeterminate calculus (Section \ref{sec:fork_side}). We then argue that the side conditions \eqref{eq:new_wrong} and $l \notin \LID \ s$ are equivalent, and that an even weaker formulation of this side condition is possible. The core of the argument is in Section \ref{sec:new_side}, in which we also describe its formalization in \operationalsemantics. The argument is concluded in Section \ref{sec:executions}, in which we describe \executions, the formalization of executions.

\subsection{Side Condition for Rule \fork}\label{sec:fork_side}

Can a revision identifier $r$ be safely allocated if one uses side condition \eqref{eq:fork_wrong}? The answer is no: this would result in indeterminacy, irrespective of the side condition on rule~\new.

What follows is a counterexample to determinacy. As a visual aid, we underline redexes $r$ of expressions $e = \E[r]$. Define the program expression
\[
P =  \big( \lambda x \ldotp \rfork \ (\rjoin \ x) \bullet ( \rjoin \ x \bullet \rfork \ \unit ) \big)  \bullet \underline{\rfork \ \unit}
\]
and consider an initial state $\{ r_1 \mapsto \langle \epsilon, \epsilon, P \rangle \}$. In what follows, we will omit the stores, because they will remain empty.
\newcommand{\negspace}{\hspace{-15mm}}
Consider the following execution trace:
\small
\[
\begin{array}{lllll}
 & \hspace{-2mm}  \{ r_1 \mapsto P \}  \\
\to_{r_1} & \hspace{-2mm} \{ r_1 \mapsto \underline{\big( \lambda x \ldotp \rfork \ (\rjoin \ x) \bullet ( \rjoin \ x \bullet \rfork \ \unit )  \big) \bullet r_2}, \\
 & \hspace{-2mm} \ r_2 \mapsto \unit \} \\
 \to_{r_1} & \hspace{-2mm} \{ r_1  \mapsto  \underline{\rfork  \ (\rjoin \ r_2)} \bullet ( \rjoin \ r_2 \bullet \rfork \ \unit ), \\
 & \hspace{-2mm} \ r_2 \mapsto \unit \} \\
\to_{r_1} & \hspace{-2mm} \{ r_1  \mapsto r_3 \bullet (\underline{\rjoin \ r_2} \bullet \rfork \ \unit ), r_2  \mapsto \unit, \\ 
& \hspace{-2mm} \ r_3  \mapsto  \underline{\rjoin \ r_2} \} \\ 
 \to_{r_1}  & \hspace{-2mm} \{ r_1  \mapsto r_3 \bullet ( \unit  \bullet  \underline{\rfork \ \unit} ),  r_3  \mapsto \underline{\rjoin \ r_2} \} \\ 
\to_{r_1} & \hspace{-2mm} \{ r_1  \mapsto r_3 \bullet ( \unit  \bullet  r_4 ),  r_3  \mapsto \underline{\rjoin \ r_2}, r_4 \mapsto \unit \}
\end{array}
\]
\normalsize
By \eqref{eq:fork_wrong}, $r_1$, $r_2$ and $r_3$ are pairwise distinct, and so are $r_1$, $r_3$ and $r_4$. But $r_2$ and $r_4$ may be equal, since $r_2$ occurred only in an expression when $r_4$ was forked. If $r_2 = r_4$, then $r_3$ performs a \join step resulting in the terminal global state
$
s = \{
r_1 \mapsto r_3 \bullet ( \unit \bullet r_4), \ 
r_3 \mapsto \unit
\}
$.
If $r_2 \neq r_4$, however, $r_3$ performs a \joinepsilon step, collapsing the global state to $\epsilon$ $\not\approx s$.

Thus, the revision calculus is nondeterminate if \eqref{eq:fork_wrong} is used. Using the side condition $r \notin \RID \ s$ invalidates the counterexample, and we will see in Section \ref{sec:determinacy} that it suffices for establishing determinacy.

The proof that \eqref{eq:fork_wrong} does not suffice as the side condition for \fork is the only proof not part of the Isabelle formalization. To formalize it, a number of operational assumptions on \texttt{subst} are needed that allow it to distribute over the constructor symbols in the second reduction step.

\subsection{Side Condition for Rule \new}\label{sec:new_side}

Can a location identifier $l$ be safely allocated if  one uses side condition \eqref{eq:new_wrong}? The answer is yes. In fact, the side conditions
\[
\label{eq:new_side2}
l \notin \bigcup \{ \doms \ L \mid L \in \ran \ s \} \tag{$SC'_\textit{new}$}
\]
\eqref{eq:new_wrong} and $l \notin \LID \ s$ all turn out to be equivalent. This is because $\LID \ L = \doms \ L$ for every $L \in \ran \ s$ is an execution invariant. This finding also implies that the side conditions for \get and \set are redundant.

To prove our finding, our first step is to formalize the operational semantics assuming the conservative formulation $l \notin \LID \ s$. Its formalization is the inductive relation \texttt{revision\_step} in theory \operationalsemantics. The notation $s \to_r s'$ henceforth corresponds to \texttt{revision\_step r s s'}.

We introduce the following definition (formalized by the two Isabelle definitions \texttt{domains\_subsume} and \texttt{domains\_\allowbreak subsume\_\allowbreak globally}):
\begin{definition}[Subsumption]
	The domains of a local state $L$ \emph{subsume} its location identifiers, denoted $\subsume \  L$, when $\LID \ L \subseteq \doms \ L$. We write $\subsume_G \ s$ for a global state $s$ when $\subsume \ L$ for all local states $L \in \ran \ s$.
\end{definition}
Our claim is thus that $\subsume_G$ is an execution invariant for global states $s$. (The direction $\doms \ L \subseteq \LID \ L$ is trivial.) We prove this by means of an inductive invariant. $\subsume_G$ is not an inductive invariant itself. The reason is rule (\textit{join}): 
\[
\begin{array}{lll}
s\llbracket r \mapsto  \langle \sigma, \tau, \E[\mathsf{rjoin} \ r'] \rangle, r' \mapsto \langle \sigma', \tau', v \rangle \rrbracket \to_r   \\
s(r \mapsto \langle \sigma, \tau \mathsf{::} \tau', \E[\mathsf{unit}]\rangle, r' \mapsto \bot)
\end{array}
\]
The two  inductive assumptions
$
\subsume \ \langle \sigma, \tau, \E[\mathsf{rjoin} \ r'] \rangle
$
and
$
\subsume \  \langle \sigma', \tau', v \rangle 
$
are not strong enough to prove the obligation
$
\subsume \ \langle \sigma, \tau \mathsf{::} \tau', \E[\mathsf{unit}]\rangle
$.
Namely, the case in which $\tau'$ maps to a value containing some $l \in \textit{Lid}$ that is subsumed \emph{only} by $\dom \ \sigma'$ cannot be proven.

To take care of rule \join, the following property is needed as well (formalized by definitions \texttt{subsumes\_accessible} and \texttt{subsumes\_accessible\_globally}):
\begin{definition}
	Let $s$ be a global state with $r, r' \in \dom \ s$. We write $\subsnap \ r \ r' \ s$ if $r' \in \RID \ (s \ r)$ implies $ \LID \ (s \ r')_\sigma \subseteq \doms \ (s \ r)$.
	If $\subsnap \ r \ r' \ s$ for all $r, r' \in \dom \ s$, then we write $\subsnap_G \ s$.
\end{definition}

We show that $\subsume_G \land \subsnap_G$ is preserved under $\to$ steps. We do not yet show that it is an inductive invariant, since that requires the formalization of notions related to executions, such as the definition of an initial state. Since the proof is a contribution of this paper, we provide a proof sketch that also serves as a high-level overview for the proof in the Isabelle formalization. 

\begin{lemma}[\texttt{step\_preserves\_$\subsume_G$\_and\_$\subsnap_G$}]
	\label{lem:inductive_inv}
	Assume that $s \to_r s'$, $\subsume_G \ s$ and $\subsnap_G \ s$. Then $\subsume_G \ s'$ and $\subsnap_G \ s'$.
\end{lemma}
\begin{proof}
	
We first establish $\subsume_G \ s'$ by a case distinction on the step $s \to_r s'$. It suffices to show $\subsume \ (s' \ r'')$ for indices $r''$ that have been updated, i.e., for which\ $s \ r'' \neq s' \ r''$. Cases \apply, \ifTrue, \ifFalse, \new, \get and \set modify only revision $r$, and case \fork in addition modifies revision $r'$. In each case, the goal is shown using calculational reasoning, requiring only the assumption $\subsume \ (s \ r)$. The proof for case \join is proven similarly, but in addition requires the assumption $\subsnap_G \ s$. Case \joinepsilon is vacuous since $s' =\epsilon$.

To show $\subsnap_G \ s'$, we make two  observations. First, $\subsnap \ r \ r \ s'$ for all $r \in \dom \ s'$ follows from $\subsume_G \ s'$ (encoded by lemma \texttt{$\subsume_G$\_imp\_$\subsnap$\_refl}). Second, if $s' \ r_1 \ {=} \ s \ r_1$ and $s' \ r_2 \ {=} \ s \ r_2$, then $\subsnap \ r_1 \ r_2 \ s'$ follows directly from $\subsnap \ r_1 \ r_2 \ s$. Hence, it suffices to show that $\subsnap \ r_1 \ r_2 \ s'$ for all \emph{distinct} $r_1, r_2 \in \dom \ s'$ with $s' \ r_1 \ {\neq} \  s \ r_1$ or $s' \ r_2 \ {\neq} \ s \ r_2$. We again proceed by case analysis on the step $s \to_r s'$:
	\begin{itemize}
		\item For each of the six local rules that modify \emph{only} the revision $r$, one must show $\subsnap \ r \ r'' \ s'$ and $\subsnap \ r'' \ r \ s'$ for arbitrary $r'' \in \dom \ s'$ with $r'' \neq r$. The reasoning in each of these six cases is very similar.
		\item Case \join is like the above, except that a case distinction on $r'' \in \RID \ \tau'$ is required for showing $\subsnap \ r \ r'' \ s'$.
		\item Case (\textit{fork}) creates two new local states at $r$ and $r'$. This creates a proof obligation for six properties, namely, $\subsnap \ r_1 \ r_2 \ s'$ for distinct $r_1, r_2 \in \{ r, r', r'' \}$, where $r''$ is some arbitrary unchanged revision.
		\item Case \joinepsilon, finally, again holds vacuously.
		\qedhere
\end{itemize}
\end{proof}

Theory \operationalsemantics ends with the definition of \texttt{revision\_step\_relaxed}. This inductive relation is identical to \texttt{revision\_step}, except that the side condition for \new is \eqref{eq:new_side2}, and the side conditions for \get and \set are omitted. Here, we will write $s \to_r' s'$ for the relation \texttt{revision\_step\_relaxed r s s'}. The proof that $\to_r$ and $\to_r'$ characterize the same transition system (given the definition of an initial state) is formalized in \executions.

\subsection{Executions}
\label{sec:executions}

Theory \executions formalizes all of the notions related to executions, described in Section \ref{sec:formal-semantics}. The set \texttt{steps} encodes the abstracted relation $\to$. To avoid confusion with the HOL symbol for logical implication, we write $s \leadsto s'$ for $s \to s'$ in the Isabelle formalization. The closure operations are defined using definitions from the Isabelle library \texttt{Transitive\_Closure}, which also liberates us from having to prove many standard (but indispensible) closure laws, such as $(x,y) \in R^* \iff \exists n. \ (x,y) \in R^n$ and $R^* \circ R^* = R^*$.

The theory proves that  every inductive invariant is an execution invariant (Isabelle lemma \texttt{inductive\allowbreak\_invariant\_\allowbreak is\_\allowbreak execution\_invariant}), and that the property
\[
\lambda s \ldotp \ \subsume_G \ s  \land \subsnap_G \ s
\]
is an inductive invariant (\texttt{nice\_ind\_inv\_is\_\allowbreak inductive\_\allowbreak invariant}). This lemma is used to prove that $(s \to_r s') = (s \to_r' s')$ for reachable states $s$ (\texttt{transition\_\allowbreak relations\_\allowbreak equivalent}), concluding the argument started in Section~\ref{sec:new_side}.
 
In addition, inductive invariance is used to show that reachability of $s$ implies that the sets $\RID \ s$ and $\LID \ s$ are finite (lemma \texttt{reachable\_imp\_identifiers\_finite}). Its proof requires similar lemmas for all the remaining structures. The result implies that a fresh identifier can always be allocated, on the assumption that \textit{Lid} and \textit{Rid} are infinite sets (lemma \texttt{reachable\_imp\_identifiers\_\allowbreak available}). While it is understandably not mentioned in the original account, it is required for formally establishing determinacy.
 
The theory ends with a proof that reachability is closed under execution, i.e., that $s \to s'$  and reachability of $s$ imply that $s'$ is reachable (\texttt{reachability\_closed\_under\_\allowbreak execution}). This lemma is a technicality required in the proof of determinacy.
  
\section{Determinacy}
\label{sec:determinacy}

Our proof of determinacy deviates from the one found in the original account. In this section we first explain and motivate the high-level differences (Section~\ref{sec:comparison}). We then explain how our proof is formalized in theory \determinacy (Section~\ref{sec:formalization}). 

\subsection{Comparison}
\label{sec:comparison}

The original proof establishes determinacy through a sequence of linearly dependent claims:
\begin{enumerate}
\item \emph{Local determinism} is established: if $s_2 \leftto_r s_1 \approx_{\alpha\beta} s_1' \to_{\alpha \ r} s_2'$, then $s_2 \approx s_2'$.\footnote{%
	Where applicable, we make the formulations in the original account formally precise. In this case, the assumption was written as $s_2 \leftto_r s_1 \approx s_1' \to_r s_2'$, which is slightly incorrect: the relation between revision $r$ in $s_1$ and revision $r$ in $s_2$ can be arbitrary.} The proof relies on the statement that ``for a fixed revision $r$, [an expression context $\E[e]$] is matched uniquely by at most one operational rule'', which we will call \emph{rule determinism}. Note that the local determinism lemma assumes, rather than infers, the existence of the step $s_1' \to_{\alpha \ r} s_2'$ which can be understood as ``mimicking'' the step $s_1 \to_r s_2$.

\item \emph{Strong local confluence} is proven: for reachable states $s_1$ and $s_1'$ with $s_2 \leftto_r s_1 \approx_{\alpha\beta} s_1' \to_{r'} s_2'$, there exist states $s_3$ and $s_3'$ such that $s_2 \to^=_{\alpha^{-1} \ r'} s_3 \approx_{\alpha\beta} s_3' \leftto_{\alpha \ r}^= s_3'$. The case where $r' = \alpha \ r$ follows from local determinism, and the case $r' \neq \alpha \ r$ is proven by a double case analysis on $s_1 \to_r s_2$ and $s_1' \to_{r'} s_2'$.

\item The relation $\to$ is lifted to a relation $\to_\abeq$ over classes of $\approx$-equivalent states, i.e., $C \toeq C'$ if there exist states $s \in C$ and $s' \in C'$ such that $s \to s'$. From strong local confluence, it follows that $C_2  \lefttoeq C_1 \toeq  C_3$ implies the existence of a class $C_4$ such that $C_2 \toeq^= C_4 \lefttoeq^= C_3$.

\item From this locally commuting property of $\to_\abeq$, it is claimed that a routine \emph{diagram tiling} \cite{bezem1998diagram} proof  establishes \emph{confluence} of $\toeq$, i.e., that $C_2  \lefttoeq^* C_1 \toeq^* C_3$ implies $C_2 \toeq^* C_4 \lefttoeq^* C_3$ for some $C_4$. The proof itself is not given.

\item Without further comment, confluence of $\to$ modulo $\approx$ is concluded from confluence of $\toeq$.

\item Determinacy of $\to$ modulo $\approx$ is subsequently obtained as a corollary.
\end{enumerate}

From a formal perspective, we first observe that item (5) is problematic. Namely, a joining reduction $C \toeq^* C'$ could be due to a noncontiguous $\to$ reduction sequence
\[
S =
\begin{array}{lllllllllllllll}
s  & \to & s_0  & & s_1 & \to & s_2 \\
& & \happrox & & \happrox & & \happrox & \\
& & s_0' & \to & s_1'  & & s_2' & \to & \cdots & \to & s'
\end{array}
\]
where $s \in C$, $s \in C'$, and $s_i \neq s_i'$ for some $i$. However, the existence of a contiguous $\to$ reduction follows from such an $S$ if equivalent states can mimic each other's steps, i.e., if whenever $s_2 \leftto s_1 \approx s_1'$, there exists an $s_2'$ such that  $s_1' \to s_2' \approx s_2$. This property, which we will call the \emph{mimicking property}, is stronger than local determinism.

During the formalization process, we first proved the mimicking property. We then realized that strong local confluence and mimicking can be applied directly in a diagram tiling proof for proving confluence of $\to$, eliminating the need to lift and unlift the relation $\to$. This simplifies items~\mbox{(3--5)} above. We also realized that the statements of local determinism and strong local confluence could be simplified: the equivalences in the sources of the divergences are not needed (e.g., the condition for local determinism becomes $s_2 \leftto_r s_1 \to_r s_2'$). This simplifies items~(1--2), which we experienced to be advantageous for the mechanization: we only have to reason about renamings (more specifically, swaps) whenever divergent nondeterministic steps are considered. Item~(6) is the same in our account.

In summary, the outline of our proof is as follows:
\begin{enumerate}
	\item Rule determinism is established.
	\item We prove our simplified statement of local determinism: if $s_2 \leftto_r s_1 \to_r s_2'$, then $s_2 \approx s_2'$. 
	\item We prove our simplified statement of strong local confluence: if $s_1$ is reachable and $s_2 \leftarrow_r s_1 \rightarrow_{r'} s_2'$, then there exist $s_3$ and $s_3'$ such that $s_2 \rightarrow_{r'}^= s_3 \approx s_3' \leftarrow_r^= s_2'$. As a technical detail, this lemma in addition requires that \textit{Rid} and \textit{Lid} are infinite sets.
	\item Independently, we prove the mimicking property.
	\item From the mimicking property and strong local confluence, confluence of $\to$ modulo $\approx$ is proven using a straightforward diagram tiling proof.
	\item Determinacy of $\to$ modulo $\approx$ is obtained as a corollary.
\end{enumerate}

\subsection{Formalization}
\label{sec:formalization}

We now explain our proof in more detail, and immediately relate it to the Isabelle formalization.

Theory \determinacy first proves nine rule determinism lemmas, one for each rule of the operational semantics. Intuitively, these lemmas state that if $s \to s'$ and $s$ matches the source state of a rule $R$, then $s'$ matches the target state of $R$. The lemma for \apply (lemma \texttt{app\_deterministic}), for instance, states that
\[
\begin{array}{ll}
s \ r  = \langle \sigma, \tau, \E[ (\lambda x. e) \ v ]   \rangle  \Longrightarrow (s \to s') = \\ (s' = s(r \mapsto \langle \sigma, \tau, \E[ [v/x] e] \rangle)).
\end{array}
\]
The lemmas for \new and \fork are deterministic up to naming only. For instance, the rule for \new (lemma \texttt{new\_pseudo\allowbreak deter\allowbreak ministic}) states that
\[
\begin{array}{ll}
s \ r  = \langle \sigma, \tau, \E[ \mathsf{ref} \ v]\rangle \Longrightarrow (s \to s') = \\
 (\exists l\ldotp l \notin \LID \ s \land s' = s(r \mapsto \langle \sigma, \tau(l \mapsto v), \E[ l ] \rangle))\text{.}
\end{array}
\]
The proofs of these lemmas follow easily from the unique decomposition lemma. The lemmas are declared as simplification rules, and are useful in the proof of local determinism.

\begin{lemma}[\texttt{local\_determinism}]\label{lem:local_determinism}
	$s_2 \leftarrow_r s_1 \rightarrow_r s_2' \Longrightarrow s_2 \approx s_2'$.
\end{lemma}
\begin{proof}
	By a case analysis on the left step $s_2 \leftarrow_r s_1$. In every case other than (\textit{new}) and (\textit{fork}), we obtain $s_2' = s_2$ by rule determinism: a case distinction on the right step is not necessary. In case (\textit{new}), we are given that $s_2 = s(r \mapsto \langle \sigma, \tau(l \mapsto v), \E[l]\rangle)$ (for $l \notin \LID \ s$), and by rule determinism,  $s_2' = s(r \mapsto \langle \sigma, \tau(l' \mapsto v), \E[l']\rangle)$ (for $l' \notin \LID \ s$). Define $\alpha = \id$ and the swap $\beta = \id(l := l', l' := l)$. It suffices to prove $\R \ \alpha \ \beta \ s_2 = s_2'$, which is derived using \textit{auto} roughly as follows. The distributive laws for renaming push the renaming inwards. The conclusions of the swap rules get matched. The assumptions of the swap rules are derived from $l \notin \LID \ s$, $l' \notin \LID \ s$ and the simplification rules for occurrences, canceling out all redundant renamings. The argument for case~(\textit{fork}) is analogous to case (\textit{new}).
\end{proof}

Our statement of strong local confluence is as follows.

\begin{theorem}[\texttt{strong\_local\_confluence}]\label{lem:slc}
	Assume that $s_1$ is reachable and that \textit{Rid} and \textit{Lid} are infinite. Then $s_2 \leftarrow_r s_1 \rightarrow_{r'} s_2' \Longrightarrow \exists s_3 \ s_3' \ldotp s_2 \rightarrow_{r'}^= s_3 \approx s_3' \leftarrow_r^= s_2'$.
\end{theorem}

The case $r = r'$ follows from Lemma \ref{lem:local_determinism}. For the $r \neq r'$ case, we conceptually follow the original proof in that we proceed by a double case analysis on the assumption $s_2 \leftto_r s_1 \to_{r'} s_2'$. This generates 81 cases, many of which are highly similar. We manage this explosion of proof obligations as follows. 

First, we prove the following  lemma which helps deal with the 36 symmetric cases:
\begin{lemma}[\texttt{SLC\_sym}]\label{lem:slc_sym}
	$\exists s_3 \ s_3' \ldotp s_2 \rightarrow_{r'}^= s_3 \approx s_3' \leftarrow_r^= s_2' \Longrightarrow
	\exists s_3 \ s_3' \ldotp s_2 \rightarrow_{r'}^= s_3 \approx s_3' \leftarrow_r^= s_2'$.
	\end{lemma}
When applied in a proof context for a case (\textit{rule})/(\textit{rule}$'$), \texttt{SLC\_sym} transforms the conclusion into its symmetric version, which at that point already has a proof.

Second, in many cases the steps commute directly. In these cases, the following lemma is used as an introduction rule:
\begin{lemma}[\texttt{SLC\_commute}]
	$s_2 \to_{r'} s_3 = s_3' \leftto_r s_2' \Longrightarrow s_2 \rightarrow_{r'}^= s_3 \approx s_3' \leftarrow_r^= s_2'$.
\end{lemma}
By applying the rule, the proof obligation is refined, which helps guide \textit{auto} and leads to understandable Isar proofs. Lemmas \texttt{join\_and\_local\_commute}, \texttt{local\_steps\_\allowbreak commute} and \texttt{local\_and\_{\allowbreak}rfork\_commute} have similar roles, refining the proof obligation even further for the commuting pairs (\textit{join})/(\textit{local}), (\textit{local})/(\textit{local}) and
(\textit{local})/(\textit{fork}), respectively.

	Finally, we only perform a case analysis on the left step $s_2 \leftto_r s_1$ in the Isabelle proof to Theorem~\ref{lem:slc}. Each of the nine cases is established by a separate lemma named \texttt{SLC\_}\textit{rule}, with \textit{rule} one of the nine rule names. These nine lemmas are proven in the order of the following proof sketch.

\begin{proof}[Proof  of Theorem~\ref{lem:slc}]
	The case distinction on the left step \mbox{$s_2 \leftto_r s_1$} generates nine cases that are proven in the following order. We use commuting diagrams to visually summarize proofs. 
	\begin{enumerate}
		\item (\textit{join}$_\epsilon$):
		Suppose revision $r$ joins a nonexistent revision $r''$ in the left step.  $s_1 \to_{r'} s_2'$ is either a (\textit{join}$_\epsilon$) step (joining some $r'''$) or not (denoted by $\overline{\textit{join}_\epsilon}$):
	\begin{center}
			\begin{minipage}{.211\textwidth}
				\centering
				\begin{tikzcd}[row sep=0.6em, column sep=3.1em]
					s_1 \arrow[rr, "r' \colon \text{join}_\epsilon (r''')"] \arrow[dd, "r \colon \text{join}_\epsilon (r'')"', swap] &  & s_2' \arrow[dd, equals] \\
					&  &  \\
					s_2 \arrow[r, equals] & \epsilon \arrow[r, equals] & \epsilon
				\end{tikzcd}
			\end{minipage}
			\begin{minipage}{.211\textwidth}
				\centering
				\begin{tikzcd}[row sep=0.6em, column sep=3.1em]
					s_1 \arrow[rr, "r' \colon \overline{\text{join$_\epsilon$}}"] \arrow[dd, "r \colon \text{join}_\epsilon (r'')"', swap] &  & s_2' \arrow[dd, "r \colon \text{join}_\epsilon(r'')", swap] \\
					&  &  \\
					s_2 \arrow[r, equals] & \epsilon \arrow[r, equals] & \epsilon
				\end{tikzcd}
			\end{minipage}
		\end{center}
	Observe that the right diagram would fail for the case $\overline{\textit{join}_\epsilon} =$ \fork if side condition \eqref{eq:fork_wrong} were used.
		
		\item (\textit{join}):
		Suppose revision $r$ successfully joins a revision $r''$ in the left step. $s_1 \to_{r'} s_2'$ either  also succesfully joins $r''$ or not (denoted by $\overline{\textit{join}(r'')}$):
		\begin{center}
			\begin{minipage}{.211\textwidth}
				\centering
				\begin{tikzcd}[row sep=0.6em, column sep=3.1em]
					s_1 \arrow[rr, "r' \colon \text{join} (r'')"] \arrow[dd, "r \colon \text{join} (r'')"', swap] &  & s_2' \arrow[dd, "r \colon \text{join}_\epsilon(r'')", swap] \\
					& &  \\
					s_2 \arrow[r, "r' \colon \text{join}_\epsilon(r'')"'] & \epsilon \arrow[r, equals] & \epsilon
				\end{tikzcd}
			\end{minipage}
			\begin{minipage}{.211\textwidth}
				\centering
				\begin{tikzcd}[row sep=0.6em, column sep=3.1em]
					s_1 \arrow[rr, "r' \colon \overline{\text{join} (r'')}"] \arrow[dd, "r \colon \text{join} (r'')"', swap] &  & s_2' \arrow[dd, "r \colon \text{join}(r'')", swap] \\
					& &  \\
					s_2 \arrow[r, "r' \colon \overline{\text{join}(r'')}"'] & s_3 \arrow[r, equals] & s_3'
				\end{tikzcd}
			\end{minipage}
\end{center}
		
		\item (\textit{local}):
		Under a (\textit{local}) step we here understand any step that is an \apply, \ifTrue, \ifFalse, \get or  \set step. The right step is a ($*$) (\textit{local}), (\textit{new}) or (\textit{fork}) step:
		\begin{center}
			\begin{tikzcd}[row sep=0.6em, column sep=3.1em]
				s_1 \arrow[rr, "r' \colon *"] \arrow[dd, "r \colon \text{local}"] &  & s_2' \arrow[dd, "r \colon \text{local}", swap] \\
				&  &  \\
				s_2 \arrow[r, "r' \colon *", swap] & s_3 \arrow[r, equals] & s_3'
			\end{tikzcd}
		\end{center}
		
		\item \new:
		Suppose the left step allocates a location identifier $l$. Either the right step also allocates $l$ or it does not (i.e., it allocates some $l' \neq l$ or is some (\textit{fork}) step):
		\vspace{1.5mm}
		\begin{center}
			\begin{minipage}{.211\textwidth}
				\centering
				\begin{tikzcd}[row sep=0.6em, column sep=3.1em]
					s_1 \arrow[rr, "r' \colon \text{new}(l)"] \arrow[dd, "r \colon \text{new}(l)"] &  & s_2' \arrow[dd, "r \colon \text{new}(l'')", swap] \\
					&  &  \\
					s_2 \arrow[r, "r' \colon \text{new}(l'')", swap] & s_3 \arrow[r, "\approx" description, no head] & s_3'
				\end{tikzcd}
			\end{minipage}
			\begin{minipage}{.211\textwidth}
				\centering
				\begin{tikzcd}[row sep=0.6em, column sep=3.1em]
					s_1 \arrow[rr, "r' \colon \overline{\text{new}(l)}"] \arrow[dd, "r \colon \text{new}(l)"] &
					& s_2' \arrow[dd, "r \colon \text{new}(l)", swap]  \\
					&   &  \\
					s_2 \arrow[r, "r' \colon \overline{\text{new}(l)}", swap] & s_3 \arrow[r, equals] & s_3' 
				\end{tikzcd}
			\end{minipage}
		\end{center}
		\item \fork:
		Finally, we consider the case where the left step is a (\textit{fork}) step. The right step is a \fork step as well. Both steps either fork the same revision identifier $r''$ or not ($r''' \neq r''$):
	\begin{center}
			\begin{minipage}{.211\textwidth}
				\centering
				\begin{tikzcd}[row sep=0.6em, column sep=3.1em]
					s_1 \arrow[rr, "r' \colon \text{fork}(r'')"] \arrow[dd, "r \colon \text{fork}(r'')"] &  & s_2' \arrow[dd, "r \colon \text{fork}(r''')", swap] \\
					& &  \\
					s_2 \arrow[r, "r' \colon \text{fork}(r''')", swap] & s_3 \arrow[r, "\approx" description, no head] & s_3'
				\end{tikzcd}
			\end{minipage}
			\begin{minipage}{.211\textwidth}
				\centering
				\begin{tikzcd}[row sep=0.6em, column sep=3.1em]
					s_1 \arrow[rr, "r' \colon \text{fork}(r''')"] \arrow[dd, "r \colon \text{fork}(r'')"] &  & s_2' \arrow[dd, "r \colon \text{fork}(r'')", swap] \\
					&  &  \\
					s_2 \arrow[r, "r' \colon \text{fork}(r''')"'] & s_3 \arrow[r, equals] & s_3'
				\end{tikzcd}
			\end{minipage}
\end{center}
	\end{enumerate}

The following table summarizes which case is addressed by which item in the given enumeration. The values for symmetric cases are grayed out and solved using Lemma~\ref{lem:slc_sym}.
	\vspace{1mm}

\begin{tabular}{l|lllllllll}
	& \rotatebox{90}{$\joinepsilon$}      & \rotatebox{90}{$\join$} & \rotatebox{90}{$\apply$}  & \rotatebox{90}{$\ifTrue$} & \rotatebox{90}{$\ifFalse$} & \rotatebox{90}{$\get$} & \rotatebox{90}{$\set$}    & \rotatebox{90}{$\new$} & \rotatebox{90}{$\fork$} \\ \hline
	$\joinepsilon$ & 1                                                 & 1                        & 1                        & 1                        & 1                        & 1                        & 1                        & 1                        & 1       \\
	$\join$        & {\color[HTML]{9B9B9B} 1}                          & 2                        & 2                        & 2                        & 2                        & 2                        & 2                        & 2                        & 2       \\
	$\apply$       & {\color[HTML]{9B9B9B} 1}                          & {\color[HTML]{9B9B9B} 2} & 3                        & 3                        & 3                        & 3                        & 3                        & 3                        & 3       \\
	$\ifTrue$      & {\color[HTML]{9B9B9B} 1}                          & {\color[HTML]{9B9B9B} 2} & {\color[HTML]{9B9B9B} 3} & 3                        & 3                        & 3                        & 3                        & 3                        & 3       \\
	$\ifFalse$     & {\color[HTML]{9B9B9B} 1}                          & {\color[HTML]{9B9B9B} 2} & {\color[HTML]{9B9B9B} 3} & {\color[HTML]{9B9B9B} 3} & 3                        & 3                        & 3                        & 3                        & 3       \\
	$\get$         & {\color[HTML]{9B9B9B} 1}                          & {\color[HTML]{9B9B9B} 2} & {\color[HTML]{9B9B9B} 3} & {\color[HTML]{9B9B9B} 3} & {\color[HTML]{9B9B9B} 3} & 3                        & 3                        & 3                        & 3       \\
	$\set$         & {\color[HTML]{9B9B9B} 1}                          & {\color[HTML]{9B9B9B} 2} & {\color[HTML]{9B9B9B} 3} & {\color[HTML]{9B9B9B} 3} & {\color[HTML]{9B9B9B} 3} & {\color[HTML]{9B9B9B} 3} & 3                        & 3                        & 3       \\
	$\new$         & {\color[HTML]{9B9B9B} 1}                          & {\color[HTML]{9B9B9B} 2} & {\color[HTML]{9B9B9B} 3} & {\color[HTML]{9B9B9B} 3} & {\color[HTML]{9B9B9B} 3} & {\color[HTML]{9B9B9B} 3} & {\color[HTML]{9B9B9B} 3} & 4                        & 4       \\
	$\fork$        & {\color[HTML]{9B9B9B} 1}                          & {\color[HTML]{9B9B9B} 2} & {\color[HTML]{9B9B9B} 3} & {\color[HTML]{9B9B9B} 3} & {\color[HTML]{9B9B9B} 3} & {\color[HTML]{9B9B9B} 3} & {\color[HTML]{9B9B9B} 3} & {\color[HTML]{9B9B9B} 4} & 5       \\ 
\end{tabular}

\qedhere
\end{proof}

From Theorem \ref{lem:slc}, we obtain the following lemma as a corollary:
\begin{lemma}[\texttt{SLC\_top\_relaxed}]
	Assume that $s_1$ is reachable and that \textit{Rid} and \textit{Lid} are infinite. Then $s_2 \leftarrow s_1 \rightarrow^= s_2' \Longrightarrow \exists s_3 \ s_3' \ldotp \ s_2 \rightarrow^= s_3 \approx s_3' \leftarrow^= s_2'$.
\end{lemma}
This version of strong local confluence is used in the diagram tiling proofs. In the visualizations of the proofs, we will label its diagram representation with the name $\text{SLC}^=$. 

To establish the mimicking property, we first prove a series of lemmas of the form $ (\alpha \ r \in \RID \ (\R \ \alpha \ \beta \ S)) = (r \in \RID \ S) $ and $ (\beta \ l \in \LID \ (\R \ \alpha \ \beta \ S)) = (l \in \LID \ S)$ for each of the structures $S$, with $\alpha$ and $\beta$ bijections. These lemmas imply that the allocation of a fresh identifier $r$ or $l$ can be directly mimicked by allocating $\alpha \ r$ or $\beta \ l$, respectively. This fact is used in the proof to the lemma below.
\begin{lemma}[\texttt{mimicking}]\label{lem:mimicking}
	If $s \to_r s'$, then $\R \ \alpha \ \beta \ s \to_{\alpha \ r} \R \ \alpha \ \beta \ s'$ for  bijections $\alpha$ and $\beta$.
	\end{lemma}
From Lemma \ref{lem:mimicking} we derive the following transitive variant, which is the version used in the diagram tiling proofs (we label its diagram representation with $\mathcal{M}^*$):
\begin{lemma}[\texttt{mimic\_trans}]\label{lem:mimic_trans}
	$s_2 \leftto^* s_1 \approx s_1' \Longrightarrow \exists s_2' \ldotp \   s_1' \to^* s_2' \approx s_2$.
\end{lemma}

Now that we have the two necessary diagrams, we follow the original account by establishing confluence modulo $\approx$ in two steps:

\begin{lemma}[\texttt{strip\_lemma}]
	Assume that $s_1$ is reachable and that \textit{Rid} and \textit{Lid} are infinite. Then $s_2 \leftto^* s_1 \to^= s_2' \Longrightarrow \exists s_3 \ s_3' \ldotp \ s_2 \to^* s_3  \approx s_3' \leftto^* s_2'$.
\end{lemma}
\begin{proof}
	By induction on the length $n$ of $s_2 \leftto^n s_1$. The Isabelle proof of the inductive step is visualized by the following diagram, in which $\twoheadrightarrow$ depicts $\to^*$:
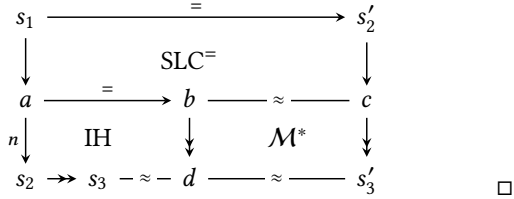
	\begin{center}
		\begin{tikzcd}[row sep=0.01em, column sep=1em]
			s_1 \arrow[dd] \arrow[rrrr, "="] &  &  & & s_2' \arrow[dd] \\
			&  & \text{SLC}^= &  &  \\
			a \arrow[dd, "n"'] \arrow[rr, "="] &  & b \arrow[dd, two heads] \arrow[rr, "\approx" description, no head] &  & c \arrow[dd, two heads] \\
			& \text{IH} &  & \mathcal{M}^* &  \\
			s_2 \arrow[r, two heads] & s_3 \arrow[r, "\approx" description, no head] & d \arrow[rr, "\approx" description, no head] &  & s_3'
		\end{tikzcd}
	\end{center}
\vspace{-6.1mm}
\qedhere
\end{proof}

\begin{lemma}[\texttt{confluence\_modulo\_equivalence}]
Assume that $s_1$ is reachable and that \textit{Rid} and \textit{Lid} are infinite. Then
$
s_2 \leftto^* s_1 \approx s_1' \to^* s_2' \Longrightarrow \exists s_3 \ s_3' \ldotp \ s_2 \to^* s_3 \approx s_3' \leftto^* s_2'
$.
\end{lemma}
\begin{proof}
By induction on the length $n$ of $s_1' \to^n s_2'$. The Isabelle proof of the inductive step is visualized by the diagram below, in which STRIP denotes the strip lemma:
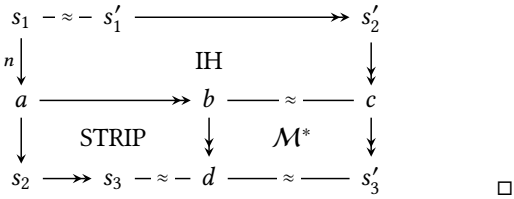
\begin{center}
	\begin{tikzcd}[row sep=0.01em, column sep=1em]
		s_1 \arrow[dd, "n"'] \arrow[r, "\approx" description, no head] & s_1' \arrow[rrr, two heads] &  &  & s_2' \arrow[dd, two heads] \\
		&  & \text{IH} &  &  \\
		a \arrow[dd] \arrow[rr, two heads] &  & b \arrow[dd, two heads] \arrow[rr, "\approx" description, no head] &  & c \arrow[dd, two heads] \\
		& \text{STRIP} &  & \mathcal{M}^* &  \\
		s_2 \arrow[r, two heads] & s_3 \arrow[r, "\approx" description, no head] & d \arrow[rr, "\approx" description, no head] &  & s_3'
	\end{tikzcd}
\end{center}
\vspace{-6.1mm}
\qedhere
\end{proof}

Finally, determinacy is obtained as a corollary by the same proof in the original account:
\begin{theorem}[\texttt{determinacy}]
Assume that $e$ is a program expression and that \textit{Rid} and \textit{Lid} are infinite. Then $e \downarrow s$ and $e \downarrow s'$ imply $s \approx s'$.
\end{theorem}
    
\section{Discussion}
\label{sec:discussion-and-related-work}

Our formalization contributes to the metatheory of concurrent revisions in two ways. First, it demonstrates that interpreting the \fork side condition  as \eqref{eq:fork_wrong} leads to nondeterminacy (Section \ref{sec:fork_side}). Second, it shows that the side condition on \new admits a weaker formulation, and that the side conditions on \get and \set are redundant (Section~\ref{sec:new_side}).

More pragmatically, what are the implications of our findings for the existing C\# \cite{burckhardt2010concurrent} and Haskell \cite{leijen2011prettier} implementations of CR? It does not seem like our counterexample in Section~\ref{sec:fork_side} is reproducible in either language. Based on the provided C\# fragments and explanations~\cite{burckhardt2010concurrent},  a ``revision identifier'' is simply a reference to an object instance of a \texttt{Revision} class. Thus, when a revision has a join pending on some object, it cannot be garbage collected, and a concurrent fork cannot replace it. Experiments in an official online environment\footnote{\url{https://rise4fun.com/Revisions}} are consistent with this analysis: join operations do not affect the hash code of a revision object $r$, and subsequent joins on $r$ return an exception. The Haskell implementation has similar characteristics, and the authors explain that a revision's data is replaced with an exception when it is joined.

Our tiling proof for determinacy clarifies that determinacy does not rely on strong local confluence only, but also on the mimicking property. While we think that one could reasonably argue that the mimicking property is too minor to mention in a paper proof, we nonetheless contend that it is valuable to have made the dependence explicit, especially if model extensions (such as a generalization of \join's merge policy) are to be considered.

We see at least three ways in which future work could meaningfully extend the formalization presented in this paper. First, the other results in the original account could also be formalized. In particular, we think that the theorem asserting the existence of a unique gca for every pair of states in a revision diagram would be interesting to formalize, since the property is important, and its paper proof relatively involved. Second, rule \join could be generalized to support custom merge functions. Third, the calculus could be extended with features that are part of the concurrent revisions project, but not yet formalized, such as support for incremental computation \cite{burckhardt2011two}. 

We think such extensions can leverage our formalization in two ways. First, all of the elementary definitions and the associated results can be directly reused, such as the unique decomposition lemma, the result that $\approx$ is an equivalence, and the lemmas required for reasoning about occurrences and renamings. Such reuse would eliminate a lot of tediousness from the formalization effort. Second, since most of our proofs are written using the structured Isar proof language, it should be quite easy to modify these proofs when, for instance, additional rules are added to the calculus: any newly generated cases can be straightforwardly integrated into the existing proofs. We consider this high degree of maintainability a great advantage of using Isabelle\slash HOL.

\subparagraph{Related Work} 
Manovit et al. \cite{manovit2006testing} developed a formal axiomatic framework and pseudorandom testing methodology for TM systems, and used it to uncover bugs in the relatively well-known Transactional memory Coherence and Consistency (TCC) \cite{hammond2004transactional} system. Cohen et al.\ \cite{cohen2008mechanical} and Doherty et al.~\cite{doherty2013towards} both developed frameworks for the formal verification of TM implementations, using the interactive theorem prover PVS. Doherty et al.~\cite{doherty2017proving} presented the first formal verification of a pessimistic (i.e., non-aborting) software transactional memory (STM) algorithm using Isabelle/HOL, extending a refinement strategy pursued in~\cite{doherty2013towards}. Abadi et al.~\cite{abadi2008semantics} developed a formal semantics for the transactional Automatic Mutual Exclusion model, and used it to study design trade-offs and errors that occur in known STM implementations.

\section{Conclusion}
\label{sec:conclusion}

We presented the first formal verification of the semantics of the concurrent revisions concurrency control model. We identified and resolved a number of ambiguities in the operational semantics, and simplified a proof of determinacy. Our paper can hopefully serve as a case study for the verification of concurrency control models, and the Isabelle\slash HOL artifact can be used as a basis for developing and verifying extensions of the concurrent revisions model.

\begin{acks}

I thank Jasmin Blanchette, Robbert van Dalen, Wan Fokkink, Hans-Dieter Hiep, Johannes H\"olzl and the anonymous reviewers for their useful discussions and/or generous feedback on preliminary versions of this manuscript. I'd also like to thank Sebastian Burckhardt for answering some of my questions and for expressing his interest in my formalization.

This paper was partially written at Centrum Wiskunde \& Informatica (CWI), Amsterdam, where it received funding from the Netherlands Organization for Scientific Research (NWO) under the COMMIT2DATA program (project No.\ 628.011.003, ECiDA). The present funding is from NWO under the Innovational Research Incentives Scheme (project No.\ VI.Vidi.192.004). The master's thesis for which the research was originally conducted was partially funded by ING.
\end{acks}

\balance 

\bibliography{bib}

\end{document}